\documentclass[prl,nobalancelastpv2.1age,twocolumn,superscriptaddress]{revtex4-1}
\pdfoutput=1

\usepackage{color,amsthm,amsmath,amsxtra,amsfonts,dsfont,graphicx,bm}
\usepackage[colorlinks=true,linkcolor=blue, citecolor=blue, urlcolor=blue, bookmarks]{hyperref}

\usepackage{amsthm}

\def\Id{{\openone}}
\newcommand{\be}{\begin{equation}}
\newcommand{\ee}{\end{equation}}
\newcommand{\bea}{\begin{eqnarray}}
\newcommand{\eea}{\end{eqnarray}}
\newcommand{\bse}{\begin{subequations}}
\newcommand{\ese}{\end{subequations}}

\usepackage{amsthm}
\theoremstyle{definition}

\theoremstyle{plain}

\newtheorem{thm}{Theorem}[section]
\newtheorem{prop}{Proposition}[section]
\newtheorem{cor}{Corollary}[section]
\newtheorem{lem}{Lemma}[section]
\newtheorem{defn}{Definition}[section]
\newtheorem{example}{Example}[section]

\newcommand{\ket}[1]{\vert#1\rangle}
\newcommand{\bra}[1]{\langle#1\vert}

\begin{document}

\title{Quantum Cellular Automata, Tensor Networks, and Area Laws}

\author{Lorenzo \surname{Piroli}}
\author{J.~Ignacio \surname{Cirac}}
\affiliation{Max-Planck-Institut f{\"{u}}r Quantenoptik,
Hans-Kopfermann-Str.\ 1, 85748 Garching, Germany}
\affiliation{Munich Center for Quantum Science and Technology, Schellingstra\ss e 4, 80799 M\"unchen, Germany}	

\begin{abstract}
Quantum Cellular Automata are unitary maps that preserve locality and respect causality. We identify them, in any dimension, with simple tensor networks (PEPU) whose bond dimension does not grow with the system size. As a result, they satisfy an area law for the entanglement entropy they can create. We define other classes of non-unitary maps, the so-called quantum channels, that either respect causality or preserve locality. We show that, whereas the latter obey an area law for the amount of quantum correlations they can create, as measured by the quantum mutual information, the former may violate it. We also show that neither of them can be expressed as tensor networks with a bond dimension that is independent of the system size.
\end{abstract}

\maketitle

Causality is a fundamental concept in Physics. It states that physical actions can not propagate in space at an arbitrary speed. In Quantum Physics, this can be mathematically captured by the notion of Quantum Cellular Automata (QCA)~\cite{farrelly2019review,arrighi2019overview}. These are the most general unitary maps between quantum states that act in discrete space (i.e., in lattices) and time, and respect causality ~\cite{richter1996ergodicity,schumacher2004reversible,arrighi2008one,arrighi2011unitarity}. They can be viewed as the quantum version of classical cellular automata, which are systems with discrete variables evolving under a local update rule. In the last years a great deal of progress has been made in the characterization of QCA. So far, complete solutions have been obtained in one ~\cite{gross2012index} and two spatial dimensions ~\cite{freedman2020classification,haah2018nontrivial,haah2019clifford,freedman2019group}.
Additionally, in the first case QCA have been identified ~\cite{chen2011Complete,Po2016Chiral, cirac2017matrix,sahinoglu2018matrix,stephen2019subsystem,hillberry2020entangled,Gong2020Classification} with Matrix Product Operators, a $1D$ version of Tensor Networks (TN), which satisfy an extra condition named simpleness~\cite{cirac2017matrix} (this has been recently extended to fermionic systems ~\cite{fidkowski2019interacting,piroli2020fermionic}). This identification connects QCA with TN, a very active area of research in many-body physics and quantum information. While most of the progress on QCA has been on unitary maps, very little is known about quantum channels representing more general physical actions ~\cite{richter1996ergodicity,brennen2003entanglement,schumacher2004reversible}, for which it is not even clear how to properly define them.

In this work, we investigate the connections between QCA and TN~\cite{verstraete2008matrix,orus2014practical}, and characterize them in terms of the amount of entanglement and correlations they can create.
First, we identify QCA in any dimension as projected entangled pair unitary (PEPU) operators that are also simple, and with a bond dimension that does not grow with the lattice size. We also show that the amount of entanglement generated by the action of a QCA is limited by an area law, similar to the one that characterizes the ground states of local Hamiltonians \cite{Reviewarealaw}.
Additionally, we analyze two natural extensions of non-unitary QCA: Causality Preserving Quantum Channels (CPQC) and Locality Preserving Quantum Channels (LPQC). While the former satisfies causality, the latter cannot create long-range correlations and fulfills an area law for the quantum mutual information. The LPQC are a strict subset of CPQC and, unlike QCA, they can not be expressed as TN with fixed bond dimension.

\begin{figure}[tbp]
	\includegraphics[width=1.0\linewidth]{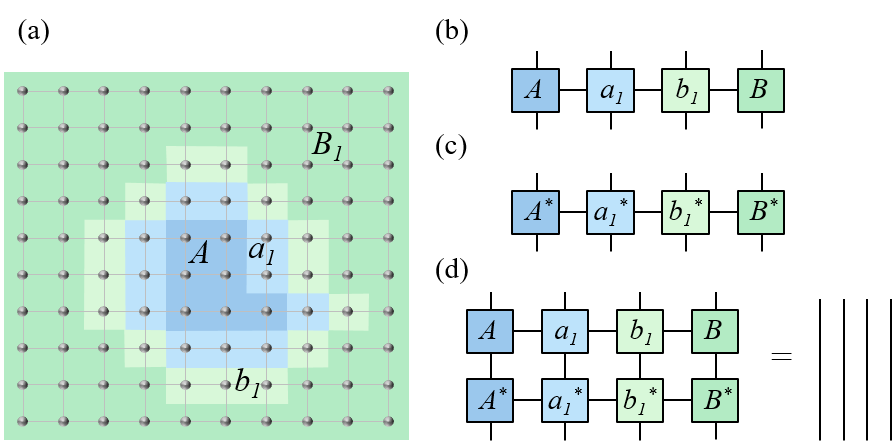}
	\caption{(a): Example of the different subsets defined in the text for $r=1$: $A$ is in dark blue and its neighborhood, $a_1$ in light blue. Their union is $\bar A_1$. $B_1$ and $b_1$ are in dark and light green respectively, and their union is $\bar B_1$. Subfigures (b), (c), (d): PEPU representation of $U$, $U^{\dagger}$, and graphical notation for the identity $UU^{\dagger}=\openone$, respectively.}
	\label{fig1}
\end{figure}

% ---------------------------------------------------------------
{\em QCA and Quantum Channels:}  We consider $N=M^{{\rm d}_L}$ qudits in a finite regular lattice in ${\rm d}_L$ dimensions. The lattice is characterized by a graph, ${\cal G}=(V,E)$, where the qudits are at the vertices $V$, which are represented by a vector $n\in \mathbb{Z}^{{\rm d}_L}$, and the edges $e_{n,m}\in E$ if $|n-m|=1$ for open boundary conditions, and similarly for periodic boundary conditions. The coordination number is $z=2{{\rm d}_L}$. The edges define a metric: the distance between to vertices, $\delta(n,m)$, is the minimum number of edges that connects them. The Hilbert space associated with the set of qudits is ${\cal H}=\otimes_{n\in V}{\cal H}_n$, where ${\rm dim}({\cal H}_n)=d$ is the physical dimension. For $r\le M/4$ and a subset $A\subset V$, we define its $r$-neighborhood, $a_r=\{n\in V\setminus A: \delta(n,A)\le r\}$, and $\bar A_r=A\cup a_r$. We further define the $r$-next-nearest neighborhood, $b_r=a_{2r}\setminus a_r$, and  $B_r=V\setminus (\bar A_r\cup b_r)$, so that $\bar B_r=B_r\cup b_r$ is the complement of $\bar A_r$ (see Fig. \ref{fig1}). We denote by $S$ all the sets $A$ such that $B_r$ is not empty. For a given $A\in S$ and $r$, the Hilbert space is decomposed as
 \be
\mathcal{H} = \mathcal{H}_A\otimes \mathcal{H}_{a_r} \otimes \mathcal{H}_{b_r} \otimes \mathcal{H}_{B_r} =
\mathcal{ H}_{\bar A_r}\otimes \mathcal{H}_{\bar B_r}\,.
 \ee
Finally, for $C\subset V$, we denote by ${\rm tr}_C$ the trace in $\mathcal{H}_C$ and by $X_C$ an operator supported on that space.

Let us now introduce a special type of quantum channels (QC) ${\cal E}$ acting on the qudits, i.e. trace-preserving completely positive maps~\cite{nielsen2002quantum}. We will denote by ${\cal E}^\dagger$ their adjoints with respect to the Hilbert-Schmidt inner product, describing the action in the Heisenberg picture.  Then
 \begin{defn}
\label{DefCPQC}
A Causality Preserving Quantum Channel (CPQC) on the lattice ${\cal G}$ with range $r$ is a channel ${\cal E}$ such that, for any $A\in S$ and $X_A$, there exists some $X_{\bar A_r}$ such that ${\cal E}^\dagger(X_A)=X_{\bar A_r}$
\end{defn}
Our definition is equivalent to that presented in Ref.~\cite{arrighi2019overview}. It states that for an observable localized at site $x$, the expectation value on the evolved state is determined by the restriction of the initial state on a neighborhood of $x$, thus justifying the name causality-preserving. When $\mathcal{E}$ is defined by a unitary operator $U$, namely ${\cal E}(X)=U X U^\dagger$ for all $X\in L(\mathcal{H})$, the set of linear operators acting on $\mathcal{H}$, we will say that the QC is unitary. Then, QCA are simply unitary CPQC. In such a case, ${\cal E}^\dagger(X)=U^\dagger XU$, and ${\cal E}^\dagger$ is still a QCA with the same range as $\mathcal{E}$~\cite{freedman2020classification}.

Before proceeding, let us mention that we could have considered more general graphs, ${\cal G}$, as long as they have no double edges nor self-loops. This would include other lattice geometries or topologies, but it would make the notation more cumbersome. Thus, in the following we will set $r=1$, drop the corresponding subindex in the sets $a,b,B$, and take $M\ge 4$~\footnote{For the lattices considered here, we can always \emph{block} (i.e. merge into blocks~\cite{cirac2017matrix_op}) $r^{{\rm d}_L}$ qudits (assuming that $\tilde M=M/r\in \mathds{N}$), and redefine the edges, so that the new lattice has $\tilde N=\tilde M^{{\rm d}_L}$ qudits, and the range of a QCA with range $r$ becomes equal to one, although the coordination number may increase}. 

We introduce now another class of QC:
\begin{defn}
${\cal E}$ is a Locality Preserving Quantum Channel (LPQC) if for any $A\in S$ and $\rho_{\bar A,\bar B}\ge 0$,
 \be
 \label{defLPQC}
 {\rm tr}_{a,b} \left[ {\cal E}(\rho_{\bar A}\rho_{\bar B})\right]=
 \frac{1}{d^N}{\rm tr}_{a,\bar B}\left[ {\cal E}(\rho_{\bar A})\right]
 {\rm tr}_{\bar A,b}\left[ {\cal E}(\rho_{\bar B})\right]\,.
 \ee
\end{defn}
This means that if we act on a product state with the quantum channel, no correlation is created between the regions $A$ and $B$. Intuitively, this corresponds to a form of localization in the Schrodinger picture, which, as we will see, represents a stronger condition than causality-preservation.

% ---------------------------------------------------------------

{\em Choi-Jamiolkowski state:} Instead of dealing with channels, it will be  useful to work with the corresponding Choi-Jamiolkowski states (CJS)~\cite{wolf2012quantum}. We associate an extra ancilla with each qudit, so that we get a copy of the lattice with vertices $V'$.  We also take $\Phi=|\Phi\rangle\langle\Phi|$, where $|\Phi\rangle = \sum_{s} |s\rangle_V\otimes |s\rangle_{V'} \in \mathcal{H}\otimes \mathcal{H}$ is an (unnormalized) maximally entangled state, and $|s\rangle=|s_1,\ldots,s_N\rangle$ is an element of the computational basis, where $s_n=1,\ldots,d$. For a channel, ${\cal E}$, its CJS is defined as $R= \left( {\cal E}_V\otimes \Id_{V'}\right) (\Phi) \in L(\mathcal{H}\otimes \mathcal{H})$, where the identity channel acts on the ancillas. It fulfills $R=R^\dagger\ge 0$, and ${\rm tr}_V(R)=\Id_{V'}$. In fact, any $R$ satisfying these conditions defines a channel, whose action is then given by ${\cal E}(\rho)={\rm tr}_{V'} (\rho_{V'}^T R)$, where the transpose is taken in the computational basis~\cite{wolf2012quantum}.

Given $A\in V$, we denote by $A'\subset V'$ the same set in the lattice of the ancillas. We can now characterize both CPQC and LPQC in terms of their CJS~\footnote{See Supplemental Material, which includes a citation to Refs.~\cite{Perez2,araki18entropy}, for a detailed proof of the statements presented in the main text.}
\begin{prop}
\label{PropCharact}
Given a channel, ${\cal E}$, for all $A\in S$ there exist $\sigma_{A,\bar A'}$ (and $\sigma_{B,\bar B'}$) such that its CJS, $R$, fulfills
\begin{description}
 \item[i] ${\rm tr}_{a,\bar B}(R)= \sigma_{A,\bar A'} \otimes \Id_{\bar B'}$ iff it is a CPQC.
 \item[ii]  ${\rm tr}_{a, b}(R)= \sigma_{A,\bar A'} \otimes \sigma_{B,\bar B'}$ iff it is a LPQC.
\end{description}
\end{prop}

The $\sigma$'s are determined by the above equations, e.g., $\sigma_{A,\bar A'}= {\rm tr}_{a,\bar B,\bar B'}(R)/d^{|\bar B^\prime|}$. This proposition expresses that the CJS of CPQC and LPQC become decorrelated if we trace some of the qudits.

% ---------------------------------------------------------------

{\em Tensor Networks:} Let us now briefly recall the TN description of quantum states, operators, and channels~\cite{verstraete2008matrix,orus2014practical}. Given a set of $N$ qudits in a graph ${\cal G}$, we associate with each vertex a tensor $A[n]$ with rank $z_n+1$, where $z_n$ is the coordination number of that vertex. We associate an index to each of the edges connecting that vertex, and the other one to the corresponding qudit. The latter is called physical index and runs from $1,\ldots,d$, and the rest are the auxiliary indices, running from $1,\ldots,D$, the bond dimension. Then, we say that 
 \be
 \label{Psi}
 |\Psi\rangle = \sum_{s} c_{s} |s\rangle
 \ee
is a TN state of bond dimension $D$ if there exist tensors $A[n]$ of that bond dimension, such that each $c_s$ can be obtained by assigning the value $s_n$ to the physical index of $A[n]$ and contracting the rest of the indices according to the lattice~\cite{verstraete2008matrix,orus2014practical}. For arbitrary lattices, they are called projected entangled pair states (PEPS). Analogously, TN can define operators and maps. For operators, we can replace $|s\rangle$ by $|s\rangle\langle s'|$, so that now the tensors $B[n]$ have two physical indices each, and for maps the tensors $C[n]$ have four. They are called PEPO (or PEPU if they are unitary) and PEPM of bond dimension $D$, respectively. Any PEPU (PEPM) has the same TN description as the PEPS (PEPO) corresponding to its CJS, and thus the same bond dimension.

The graphical representation of TN~\cite{verstraete2008matrix,orus2014practical} consists in replacing each tensor by a box, each index by a line, and contraction of indices by identifying the corresponding lines. For a graph, ${\cal G}$, PEPS, PEPO, PEPM are thus represented by the same graph where each of the vertices is replaced by a tensor that has one, two and four lines with open ends, respectively, and otherwise they are connected according to the edges. We can block tensors to represent blocks of qudits. For instance, the representation of two PEPU, $U$ and $U^\dagger$, acting on sets $AabB$ is shown in Fig. \ref{fig1}(b,c). We have written in each box the name of the set where the tensor acts, and used an asterisk to specify that the tensor is transposed and complex conjugated. Figure \ref{fig1}(d) represents $UU^\dagger=\Id$, where the multiplication is read from bottom to top. The bond dimension for the tensor corresponding to $A$ is $D^{z_A}$, where $z_A$ is the number of edges connecting $A$ with its neighborhood $a$, and the physical dimension is $d^{|A|}$. We can now define a notion that was introduced in \cite{cirac2017matrix}.

\begin{defn}
We say that a PEPU is simple if for any $A\in S$
 \be
 \label{simple1}
\raisebox{-24pt}{\includegraphics[height=5em]{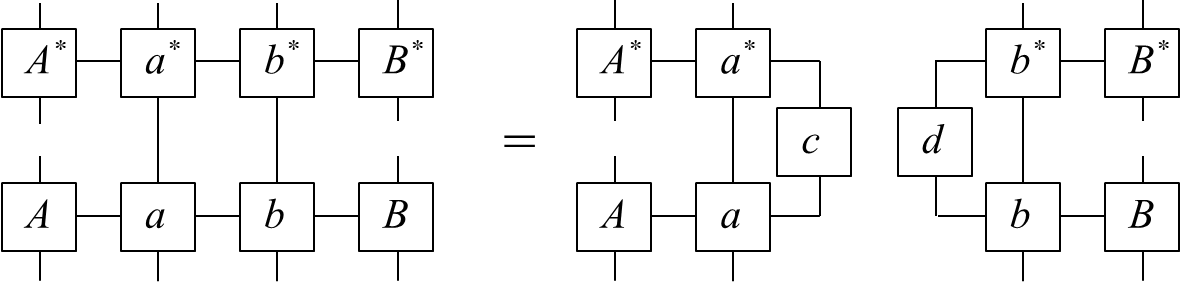}}
 \ee
where
  \be \label{simple2}
\raisebox{-24pt}{\includegraphics[height=5em]{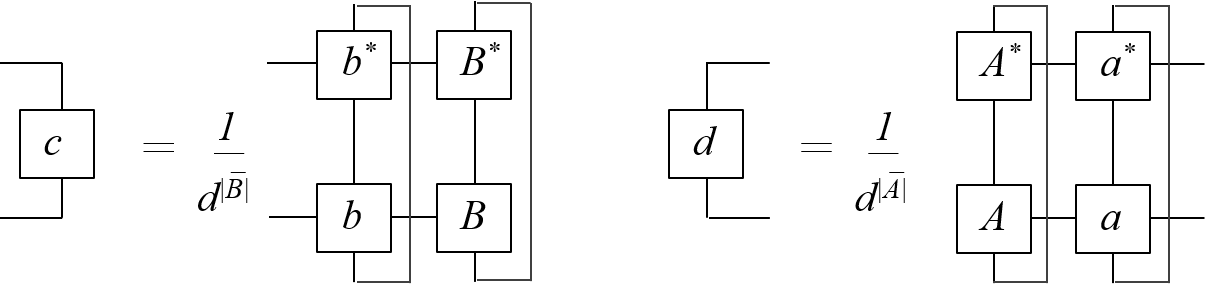}}
 \ee
\end{defn}

% ---------------------------------------------------------------

{\em Quantum Cellular Automata:} We establish now the connection between QCA and PEPU as well as with LPQC.
\begin{thm}
\label{ThmuQCA}
Given a unitary channel acting as ${\cal E}(\rho)=U \rho U^\dagger$ on the qudits of a lattice , the following statements are equivalent:

\begin{description}
\item[i] ${\cal E}$ is a QCA (namely, a CPQC).
\item[ii] ${\cal E}$ is a LPQC.
\item[iii] $U$ can can be represented by a simple PEPU, where $D$ only depends on $d$, ${{\rm d}_L}$ and $z$.
\end{description}
\end{thm}

While all unitaries can be represented by PEPU, the last equivalence establishes that for a QCA this can be done efficiently, namely with a bond dimension that does not depend on $N$. This has strong implications on the amount of entanglement that a unitary $U$ associated with a QCA can create between any two regions. If one applies $U$ to a pure product state $\ket{\Psi}$, then, the entanglement of $U\ket{\Psi}$ between any set $A$ and the rest is $\le D^{|\partial A|}$, where $|\partial A|$ is the number of edges between $A$ and $a$. This gives rise to an area law; to see that, we have to consider a sequence of QCA, $S_E=\{{\cal E}_M\}_{M=4}^\infty$, each acting on a lattice of $M^{{\rm d}_L}$ qudits. Furthermore, we denote by $E(A:A^c)$ the entanglement entropy~\cite{nielsen2002quantum} between the qudits in $A\subset V$ and its complement $A^c=V/A$, and by $\partial A$ their boundary.

\begin{defn}
A sequence of QCA obeys an area law if for all $A\subset V$, the state obtained by applying any of the QCA to any pure product state fulfills $E(A:A^c)\le c |\partial A|$, where $c$ is a constant independent of $M$.
\end{defn}
Thus, Theorem \ref{ThmuQCA} immediately implies that:
\begin{cor}
Any sequence of QCA satisfies an area law.
\end{cor}

% ---------------------------------------------------------------
{\em General CPQC:} General CPQC possess very different properties than QCA. For instance, the set of CPQC acting on qudits in a lattice is convex. Note that this is not true for LPQC. Furthermore, whereas for QCA and any region $A\in S$,
 \be
 \label{factor}
 {\cal E}^\dagger(X_A Y_{B}) = {\cal E}^\dagger(X_A) {\cal E}^\dagger(Y_B)\,,
 \ee
with ${\cal E}^\dagger(X_A) =X_{\bar{A}}$, ${\cal E}^\dagger(Y_B)=Y_{\bar B}$, this is not necessarily true for CPQC.

Any channel (unitary or not) can be written in terms of a unitary operator through the Stinespring dilation~\cite{wolf2012quantum}. In particular, we can consider the channel ${\cal E}$ built out of a QCA, ${\cal E}_u: L(\mathcal{H}\otimes \mathcal{H})\to L(\mathcal{H}\otimes \mathcal{H})$ as
 \be
 \label{Stinespring}
 {\cal E}(\rho)={\rm tr}_{V'}\left[{\cal E}_u\left(\rho \otimes (|1\rangle\langle 1|)^{\otimes N}\right)\right]\,,
 \ee
where $|1\rangle$ is a state of the ancilla qudits~\footnote{In order to define the range of the QCA $\mathcal{E}_u$, one also needs to specify the sets of edges in the lattice made of the physical and ancillary vertices. Here we define it in the most natural way, by connecting each physical site with the corresponding ancilla.}. Let us now introduce three other sets of channels:
\begin{defn}
We define fQC as the set of CPQC fulfilling the factorization condition \eqref{factor}, while tnQC as the set of CPQC whose CJS has a PEPO description (with bond dimension bounded by a function of $d$, ${\rm d}_L$ and $z$, but not of $M$). Finally, we denote by dQC the set of CPQC that are obtained by a Stinespring dilation in terms of a QCA [that is, ${\cal E}_u$ in Eq.~\eqref{Stinespring} is a QCA].
\end{defn}

Let us give some illustrative examples. We take $d=2$, i.e. qubits, with $\{|s_n\rangle\}_{s_n=0}^1$ the local computational basis, and $\sigma_n^\alpha$ the Pauli operators.

\begin{example}
\label{Example1}
{\em A channel that is a tnQC but not a LPQC}. Let us define
 \be
 {\cal E}(\rho)= \frac{1}{2} \left[\rho + \left(\sigma^x\right)^{\otimes N} \rho \left(\sigma^x\right)^{\otimes N}\right]\,.
 \ee
${\cal E}$ is a convex combination of two tnQC with bond dimension $D=1$, and thus a tnQC with bond dimension $D=2$. Furthermore, it is also a CPQC, since it is a convex combination of two CPQC. However, it is not a LPQC since it does not satisfy Proposition \ref{PropCharact}.
\end{example}

\begin{example}
\label{Example2}
{\em A set of channels that are LPQC but not tnQC}. Let us consider the state (\ref{Psi}), where $s_n=0,1$ and each qubit $n=(n_1,n_2,\ldots,n_{{\rm d}_L})$ is maximally entangled with the qubit $n'=(n_1',n_2,\ldots,n_{{\rm d}_L})$, where $|n_1'-n_1|=M/2$ with $M$ even. Let us define $R= \Id_V\otimes \Id_{V'} / 2^N + S$ where
 \be
 S =  k_N\sum_s c_s \left[\otimes_{n=1}^N(\sigma_n^x\otimes \sigma_{n'}^x)^{s_n} (\sigma_n^z\otimes \sigma_{n'}^z)^{1-s_n}\right]\,.
 \ee
Choosing $k_N$ so that $||S||_\infty \le 1/2^N$, we have $R\ge 0$, and tracing any system or ancilla qudit we get ${\rm tr}_n(S)={\rm tr}_{n'}(S)=0$. Thus, ${\rm tr}_{V}(R)=\Id_{V'}$ and therefore $R$ is a valid CJS that defines a channel, ${\cal E}_M$, for each $M$. Furthermore, $R$ fulfills the conditions of Proposition \ref{PropCharact}, and it is therefore a LPQC. However, we claim that it does not admit a TN representation with finite bond dimension. Indeed, the latter is true iff $S$ can be represented as a PEPO with finite bond dimension.  But the (unnormalized) state $ |\Psi\rangle = \sum_s c_s\ket{s}$ is such that the rank of the reduced state in a hypercube of side $L<M/2$ is $d^{\left(L^{{\rm d}_L}\right)}$, so that its PEPS representation has a bond dimension that increases exponentially with $M$. But any 
PEPO representation of $S$ can be interpreted as PEPS for $\ket{\Psi}$ with the same bond dimension. We conclude that $R$ cannot be represented by a PEPO with bond dimension independent of $M$.
\end{example}

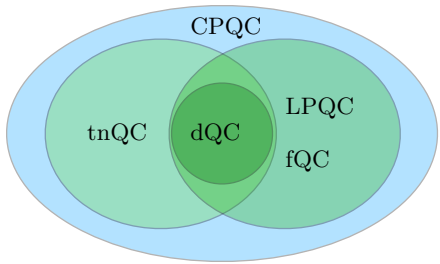
\begin{figure}
	\includegraphics[scale=0.25]{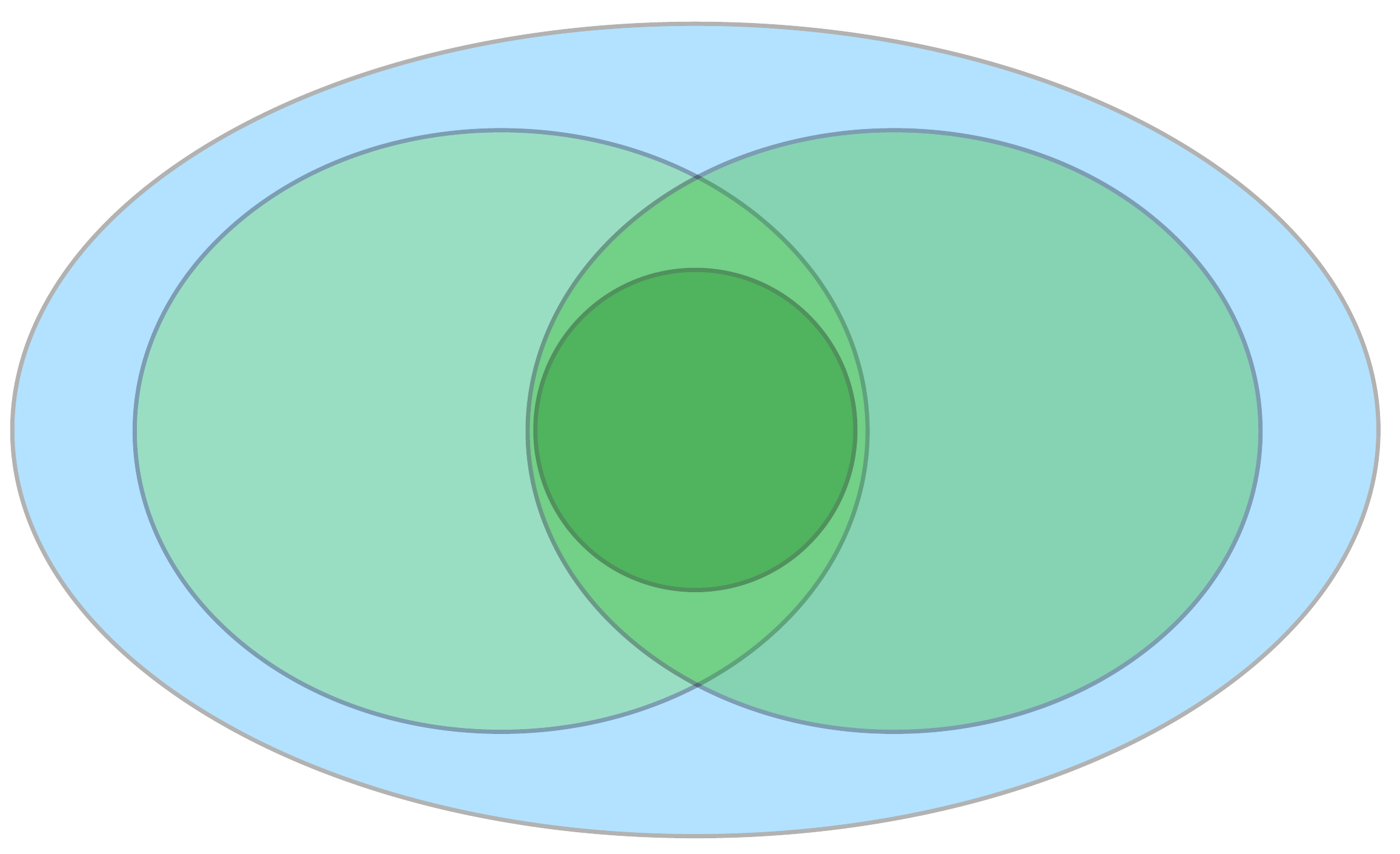}
	\put(-135,48){tnQC}
	\put(-60,58){LPQC}
	\put(-60,38){fQC}
	\put(-96,88){CPQC}
	\put(-96,48){dQC}
	\caption{Euler diagram for the class of channels defined in the main text, representing the statement of Theorem~\ref{thm2}.}
	\label{fig:venn}
\end{figure}

We are now in the position to formulate the following:
\begin{thm}
\label{thm2}
For any of the considered lattices, $dQC\subset fQC = LPQC \subset CPQC$. Furthermore,  $dQC \subset tnQC$ and $tnQC\neq LPQC$, where all inclusions are strict  (see Fig.~\ref{fig:venn}).
\end{thm}
\noindent Note that this theorem does not say whether the intersection between tnQC and LPQC coincides with the set of dQC or simply contains it. This remains an interesting open question.

Finally, let us discuss an area law for the classes of QC defined above. As irreversible QC will typically create mixed states out of pure ones, rather than talking about the entanglement it is more appropriate to investigate the amount of correlations that can be created. The relevant measure for this is the mutual information: given a state, $\rho$, in a qudit lattice, a subset of qudits, $A\in V$, and its complement, $A^c=V/A$, the mutual information is $I(A:A^c)=S_A+S_{A^c}-S_V$, where $S_A$ is the von Neumann entropy of the reduced state of the qudits in $A$~\cite{nielsen2002quantum}. For Gibbs states of local Hamiltonians or for PEPO it is known that the mutual information obeys an area law~\cite{Wolf2008Area}. This motivates the following definition:

\begin{defn}
A  sequence of QC obeys an area law if for all $A\in V$, the state obtained by applying any of the QC to any product state fulfills $I(A:A^c)\le c |\partial A|$, where $c$ is a constant independent of $M$.
\end{defn}

We can now state our third main result:

\begin{thm}
\label{arealawmutual}
Any  sequence of LPQC obeys an area law.
\end{thm}

Finally, we show that causality is not enough to bound the amount of correlations that can be created when acting on a product state:
\begin{example}\label{ex:mutual_info}
	Let us consider the dephasing channel acting on two qubits, n,m:
	\be
	{\cal E}_{n,m}(\rho)= \frac{1}{2}\left[\rho + (\sigma^z_n\otimes \sigma^z_m) \rho (\sigma^z_n\otimes \sigma^z_m)\right]\,,
	\ee
and define the channel ${\cal E}=\otimes_{n\in V_1} {\cal E}_{n,n+e}$ where $V_1$ contains all $n\in V$ with $n_1\le M/2$, and $e=(M/2,0,0,\ldots,0)$. ${\cal E}$ is a convex combination of Pauli channels, and thus a CPQC. However, the CJS is $R = \otimes_{n\in V_1} \rho_{n,n',n+e,n'+e}$. The mutual information between $(n,n')$ and $(n+e,n'+e)$ is one. Taking into account that the mutual information is additive under tensor product, we conclude that for a hypercube of side $L<M/2$, it is $L^{{\rm d}_L}$.
\end{example}

{\em Conclusions:} We have investigated the connections between QCA, TN, and generation of quantum entanglement and correlations. We have shown that QCA can be efficiently represented by TN, implying an area law for the entanglement entropy that they generate. We have explored the implications of causality and locality for irreversible QC, proving that only the former provides a constraint on the amount of quantum correlations that can be created. Still, even LPQC can not be represented efficiently via TN.

Our work opens up several questions and possibilities. The identification of QCA with PEPU allows one to use the established techniques based on TN for numerical simulations of their action~\cite{schollwock2011density,ran2020Contration}.  This also gives us a very natural framework to investigate the classification of (symmetry-protected) topological (SPT) phases for QCA~\cite{Gong2020Classification} in higher dimensions, with possible implications for the classification of Floquet SPT phases~\cite{Po2016Chiral,Roy2017Roy,Potter2017Dynamically,harper2020topology}. Additionally, QCA inherit the holographic principle of PEPS~\cite{Cirac2011Entanglement}, which can also be used for their classification. Let us also mention some questions that our work immediately raises. Given that CPQC constitute a convex set, perhaps they can be obtained as the convex hull of either dQC or, more generally, LPQC. A solution to this problem would give us a very useful characterization of this set. In turn, this might be important in order to study equivalence classes of CPQC under smooth deformations. 

%%%%%%%%%%%%%%%%%
{\em Acknowledgments:}
We thank Alex Turzillo for discussions. LP acknowledges support from the Alexander von Humboldt foundation. JIC acknowledges support by the EU Horizon 2020 program through the ERC Advanced Grant QENOCOBA No. 742102, and from the DFG (German Research Foundation) under Germany’s Excellence Strategy - EXC-2111 - 390814868.

\bibliography{./bibliography}

\newpage

% ==============================================================
\section*{Supplementary Material}

Here we will prove the results stated in the main text. In all the theorems and propositions, we use statements like ``for all $A\in S$", or ``for all $X_A$", or ``there exists a traceless $Y_{\bar B}$". In order to ease the reading, we will omit those statements when formulating the proofs whenever there is no room for confusion. We start with a characterization of QCA in the Schr\"odinger picture.

\begin{lem}
A channel ${\cal E}$ is a CPQC iff for all $A\in S$, $\sigma_{A},\rho_{\bar A}\ge 0$ and $Y_{\bar B}=Y_{\bar B}^\dagger$, with ${\rm tr}(Y_{\bar B})=0$,
 \be
 \label{characQCA}
 {\rm tr}\left[ \sigma_A {\cal E}( \rho_{\bar A} Y_{\bar B})\right]=0\,.
 \ee
\end{lem}

\begin{proof}
(if) We will use the Definition \ref{DefCPQC} of a CPQC. We can always write $X_{A}=\sigma_{A}^{1}-\sigma_{A}^{2}+i \sigma_{A}^{3}-i \sigma_{A}^{4}$, where all $\sigma^i_A\ge 0$. Thus, it is enough to show that for any $\sigma_A\ge 0$, ${\cal E}^\dagger(\sigma_A)$ is supported in ${\bar A}$. For all $\rho_{\bar A}$ and traceless $Y_{\bar B}$ we have $0={\rm tr}[{\cal E}^\dagger(\sigma_A) \rho_{\bar A} Y_{\bar B}]={\rm tr}_{\bar B}(Z_{\bar B} Y_{\bar B})$, so that $Z_B={\rm tr}_{\bar A}[{\cal E}^\dagger(\sigma_A) \rho_{\bar A}]\propto \Id_{\bar B}$, and thus ${\cal E}^\dagger(\sigma_A)$ is supported on ${\bar A}$.

\noindent
(only if) We have ${\rm tr}[\sigma_A {\cal E}(\rho_{\bar A} Y_{\bar B})]={\rm tr}[{\cal E}^\dagger(\sigma_A) \rho_{\bar A} Y_{\bar B}]=
{\rm tr}[X_{\bar A} \rho_{\bar A} Y_{\bar B}]\propto {\rm tr}_{\bar B}(Y_{\bar B})=0$.
\end{proof}

\noindent
{\bf Proposition} \ref{PropCharact}:

\begin{proof}
(i) The statement becomes trivial by noticing that
 \[
{\rm tr}[X_A {\cal E}(\rho_{\bar A}Y_{\bar B})]=
 {\rm tr}_{A,\bar A'}\left[X_A \rho_{\bar A'}^T
 \left({\rm tr}_{\bar B'}(Y_{\bar B'}^T
 \sigma_{A,\bar A',\bar B'})\right]\right)
 \]
where  $\sigma_{A,\bar A',\bar B'}={\rm tr}_{a,\bar B}(R)$.

\noindent (ii) Let us define
 \bse
 \bea
 \alpha_{A,B}&=&{\rm tr}_{a,b}\left[{\cal E}(\rho_{\bar A}\rho_{\bar B})\right],\\
 \beta_{A}&=&{\rm tr}_{a,\bar B}\left[{\cal E}(\rho_{\bar A})\right],\\
 \gamma_{B}&=&{\rm tr}_{\bar A,b}\left[{\cal E}(\rho_{\bar B})\right],
 \eea
 \ese
and $\sigma_{A,\bar A',B,\bar B'}={\rm tr}_{a,b}\left(R\right)$. Then, the statement becomes trivial by noticing that
 \begin{eqnarray*}
 \alpha_{A,B}&-&\frac{1}{d^N}\beta_{A}\gamma_B = {\rm tr}_{\bar A',\bar B'} \left[\rho^T_{\bar A}\rho^T_{\bar B}
 \sigma_{A,\bar A',B,\bar B'} - \right. \\
&& \frac{1}{d^N} {\rm tr}_{B}\left( \rho^T_{\bar A} \sigma_{A,\bar A',B,\bar B'}\right) {\rm tr}_{A}\left( \rho^T_{\bar B} \sigma_{A,\bar A',B,\bar B'}\right)
  \left. \right]\,.
 \end{eqnarray*}
\end{proof}

\noindent
{\bf Theorem} \ref{ThmuQCA}:

\begin{proof}
(ii$\Rightarrow$ i) We will show that any LPQC is a CPQC. Let us assume that ${\cal E}$ is a LPQC. We will show that it fulfills \eqref{characQCA}. Indeed, let us consider $A,X_A,\rho_{\bar A}\ge 0$ and a traceless $Y_{\bar B}=Y_{\bar B}^\dagger$. We can always write $Y_{\bar B}=\rho_{\bar B}-\tilde \rho_{\bar B}$, with $\rho_{\bar B}, \tilde \rho_{\bar B}\geq 0$. Then using \eqref{defLPQC}, and ${\rm tr}\rho_{\bar B}={\rm tr}\tilde \rho_{\bar B}$, Eq.~\eqref{characQCA} immediately follows.

\noindent
(iii$\Rightarrow$ ii) Using the fact that the PEPU is simple [Eqs.~(\ref{simple1}), (\ref{simple2})], we immediately have
 \be
 {\rm tr}_{a,b}(U \rho_{\bar A}\rho_{\bar B} U^\dagger ) =\frac{1}{d^{N}}
 {\rm tr}_{a,\bar{B}}( U \rho_{\bar A}U^\dagger )
 {\rm tr}_{\bar{A},b}(U \rho_{\bar B} U^\dagger )\,.
 \ee

(i$\Rightarrow$ ii) We do not really need to prove this, since it follows from the other implications. However, we will use it in the proof below and, additionally, this will serve as a proof for a piece of Theorem \ref{thm2}. This is why we will only use that [cf. (\ref{factor})]
 \be
 {\cal E}^\dagger (X_A Y_B) = {\cal E}^\dagger (X_A) {\cal E}^\dagger (Y_B)\,,
 \ee
which is obvious for QCA. For any $\rho_{\bar A},\rho_{\bar B}\ge 0$ let us denote by
 \be
 \sigma_{A,B} = {\rm tr}_{a,b}[{\cal E}(\rho_{\bar A}\rho_{\bar B})]\,.
 \ee
For any $X_A,Y_B$ we have
 \bea
 \label{simple3}
{\rm tr}_{A,B}(X_A Y_B \sigma_{A,B}) &=& {\rm tr}({\cal E}^\dagger (X_A)
 {\cal E}^\dagger (Y_B) \rho_{\bar A}\rho_{\bar B}) \nonumber\\
 &=&  {\rm tr}_{\bar{A}}[\tilde X_{\bar A} \rho_{\bar A}]
 {\rm tr}_{\bar{B}}[\tilde Y_{\bar B} \rho_{\bar B}]\,,
 \eea
where $\tilde X_{\bar A}={\cal E}^\dagger(X_A)$ and
$\tilde Y_{\bar B}={\cal E}^\dagger(Y_B)$.

(i$\Rightarrow$ iii) We will first show that $U$ is a PEPU with finite bond dimension. We denote by $|\Psi\rangle$ its (pure) CJS
 \be
 |\Psi\rangle = (U\otimes \Id) |\Phi\rangle\,.
 \ee
We define $Q_{n}= \Id-(1/d)|\Phi\rangle_n\langle \Phi|$ (where $(1/d)|\Phi\rangle_n\langle \Phi|$ is the projector onto the maximally entangled state between qudit $n$ and its ancilla) and 
\be
 \tilde H= (U\otimes\Id) \left[\sum_n Q_n \right] (U\otimes\Id)^\dagger = \sum_n \tilde Q_n\,,
 \label{eq:hamiltonian}
 \ee
with $\tilde Q_n= (U\otimes \Id) Q_n (U\otimes \Id)^\dagger$. These operators are local, since $U$ is a QCA, and mutually commute, $[\tilde Q_n,\tilde Q_m]=U[Q_n,Q_m]U^\dagger =0$. Furthermore, $\ket{\Psi}$ is the unique ground state of the frustration free Hamiltonian $\tilde H$.  This is because  $\tilde H$ has the same spectrum of $H=\sum_n Q_n$, and the ground state of the latter is clearly unique, since $Q_n$ is a projector with rank $d^2-1$. Then we can use the argument of Ref.~\cite{Perez2} to show that it is a PEPS with a finite bond dimension. The idea is that $\ket{\Psi}$ can be prepared by projecting a random state in the ground state
 \be
|\Psi\rangle \propto \prod_n (\Id-\tilde Q_n) \bigotimes_{m}|\alpha_m\rangle\,,
\label{eq:projection}
 \ee
where $|\alpha_m\rangle$ is any state of the qudit at site $m$ and its ancilla. Since the $\Id-\tilde Q_n$ are local and thus can be decomposed as sum of operators acting on a small region, each of the projectors creates a tensor around one region. But if $\ket{\Psi}$ can be represented by a TN with a given bound dimension, then so can $U$. Note that the bond dimension is independent of $N$. To see this, note that $T_n=(\Id-\tilde Q_n)$ acts non-trivially only on $z+1$ sites, and that the number of operators acting simultaneously on a given pair of neighboring qudits $(n,n^\prime)$ (and the corresponding ancillas) is $2$, $T_n$, and $T_{n^\prime}$. These are the only operators that modify the bond dimension of the link connecting $n$ and $n^\prime$. Thus, $D$ is clearly independent of $N$.

Finally, the fact that the PEPU is simple, immediately follows from (\ref{simple3}).
\end{proof}

It is instructive, as an example, to compute explicitly an upper bound for the bond dimension $D$ for a square lattice, i.e. ${\rm d}_L=2$ (and arbitrary local physical dimension $d$). As usual, we assume the QCA has $r=1$, and coordination number $4$, so that $T_n=\Id-\tilde Q_n$ will act on the $n$-th qudit, its $4$ nearest neighbors, and the corresponding ancillas. Thus, $T_n=\Id-\tilde Q_n$ can be represented as a plaquette with $5$ incoming and outcoming legs, each associated with a Hilbert space of dimension $d^2$ (corresponding to one physical local system and one ancilla). We can then enumerate the legs, and decompose the plaquette as a $1D$ matrix product operator, with bond dimension $D\leq d^8$. Now, by ``bending'' some of the legs, we can cast this in the form of a PEPO, where the four ``outer'' sites are only connected to the central one, labeled by $n$. The global PEPO, corresponding to $\prod_n T_n$ is finally obtained by patching together the local ones. Since the bond dimension between neighboring sites $n$, $n^\prime$ is only modified by the action of $T_n$, $T_{n^\prime}$,  it is easy to see that the above procedure can be carried out in such a way that $D\leq d^{16}$.

Note that the proof of Theorem~\ref{ThmuQCA} applies to generic QCA, not necessarily displaying translation symmetry. However, in the case $U$ is translationally invariant, the argument could be simply adjusted to recover a translationally-invariant PEPO description. In this case, we can replace the arbitrary product state in Eq. ~\eqref{eq:projection} with a state $\ket{\phi}^{\otimes N}$, which is not annihilated by $\prod_n(\openone-\tilde{Q}_n)$, where $\ket{\phi}$ is some state of the single-site qudit and the corresponding ancilla. This is possible as long as there exists a single qudit operator $T$ such that ${\rm tr}(UT^{\otimes N}) \neq 0$. If this is not the case, we could either block spins so that now $T$ can act on more qudits or apply the projectors to a PEPS with small bond dimension (instead of a product state), so that the final PEPS has still finite bond dimension. Finally, we mention that there is an alternative proof that QCA are tnQC. This is based on the fact that, given a QCA $U$, the operator $U\otimes U^{\dagger}\in L(\mathcal{H}\otimes \mathcal{H})$ can be represented as a quantum circuit in the doubled Hilbert space $\mathcal{H}\otimes \mathcal{H}$~\cite{arrighi2011unitarity,gross2012index}. The  idea is then to represent such a circuit as a PEPU, and to take a partial expectation value with respect to a product state in the second system.\\

\noindent
{\bf Theorem} \ref{thm2}

\begin{proof}
$LPQC \subset CPQC$: This has been proven already in (ii $\Rightarrow$i) in Theorem \ref{ThmuQCA}. That the inclusion is strict is clear from Example \ref{Example1}.

\noindent
$LPQC= fQC$: This follows immediately from
 \begin{eqnarray*}
 &&{\rm tr}\left[ [{\cal E}^\dagger (X_A Y_B)-{\cal E}^\dagger (X_A)
 {\cal E}^\dagger ( Y_B)] \rho_{\bar A} \rho_{\bar B}\right] =\\
 && {\rm tr}\left[ X_A Y_B {\cal E} (\rho_{\bar A} \rho_{\bar B})\right]
 - {\rm tr}_{\bar A}\left[ X_A  {\cal E} (\rho_{\bar A})\right]
 {\rm tr}_{\bar B}\left[ Y_B  {\cal E} (\rho_{\bar B})\right]\,.
 \end{eqnarray*}

\noindent
$dQC \subseteq tnQC$: It automatically follows from the fact that QCA are tnQC and that tracing does not change this fact.

\noindent
$dQC\subseteq LPQC$: It is immediate from Eq.~\eqref{Stinespring} and the fact that ${\cal E}_u$ is a QCA.

\noindent
$tnQC \ne dQC \ne LPQC$: This follows from Examples~\ref{Example1},~\ref{Example2}.

\end{proof}

\noindent
{\bf Theorem} \ref{arealawmutual}

\begin{proof}
We will prove this here for the (normalized) CJS corresponding to the channel. For the action of the channel on any product state, mixed or not, the same argument trivially applies. This proof is based on Proposition \ref{PropCharact} and the following property of the mutual information. For any sets $x,y,z$, $I(x:yz),I(xy:z)\le I(x:z) + 2 D_y$, where $D_y$ is the logarithm of the dimension of the Hilbert space corresponding to $y$. To see this, note that, using the Araki–Lieb triangle inequality~\cite{araki18entropy} we have $S(AaB) \geq  S(AB)-S(a)$, while, using subadditivity, $S(a)+S(A)\geq S(Aa)$. Then, $I(Aa:B) = S(B)+S(Aa)-S(AaB) \leq S(B)+S(A)+ S(a)+S(a)-S(AB) = I(A:B) + 2S(a)\leq  I(A:B) + 2D_a$, where we used that $S(a) \leq  D_a$. Applying this to the (normalized) CJS of a LPQC and any $A\in S$, we have
 \bea
 I(\bar A \bar A':\bar B \bar B')&\le&  I(A \bar A':B \bar B') + 2( |a|+ |b|)|\log_2(d)\nonumber\\
	&= &2( |a|+ |b|)|
	(d)\,.
 \eea
since $I(A \bar A':B \bar B')=0$, due to Proposition \ref{PropCharact}(ii). 

\end{proof}

\noindent{\bf Example} \ref{ex:mutual_info}
\ \\

\noindent Finally, we provide a few additional details regarding the construction of Example~\ref{ex:mutual_info}, showing explicitly that the mutual information between $(n,n')$ and $(n+e,n'+e)$ is one. To this end, we denote by $\ket{I}_n=(\ket{0}_n\ket{0}_{n^\prime}+\ket{1}_n\ket{1}_{n^\prime})/\sqrt{2}$ the normalized maximally entangled state between $n$  and $n^\prime$, and $\ket{z}_n=(\sigma^z_n\otimes \openone_n^\prime )\ket{I}$. Then
\begin{align}
{\cal E}_{n,n+e}(\rho_{n,n',n+e,n'+e})=&\frac{1}{2}\left[\ket{I}_n\bra{I}\otimes \ket{I}_{n+e}\bra{I} \right.\nonumber\\
+&\left. \ket{z}_n\bra{z}\otimes \ket{z}_{n+e}\bra{z}\right]\,.
\end{align}
Using that $\ket{I}$ and $\ket{z}$ are orthogonal, we obtain $S(n,n^\prime, n+e,n^\prime+e)=1$. Next, tracing over $n$,$n^\prime$, we have
\begin{align}
{\rm tr}_{n,n^\prime}[{\cal E}_{n,n+e}(\rho_{n,n',n+e,n'+e})]=&\frac{1}{2} \ket{I}_{n+e}\bra{I} \nonumber\\
	+& \frac{1}{2}\ket{z}_{n+e}\bra{z}\,,
\end{align}
so that $S(n+e,n^\prime+e)=1$. In the same way, we can show $S(n,n^\prime)=1$, finally implying $I(n,n^\prime: n+e,n^\prime+e)=1$.

\end{document}